\documentclass[11pt,a4paper]{article}
\usepackage{amsmath,amssymb,centernot, amsthm}
\usepackage{fullpage}
\usepackage{appendix}
\usepackage{graphicx}
\usepackage{latexsym}

\newtheorem{thm}{Theorem}[section]

\newtheorem{lem}[thm]{Lemma}

\newtheorem{defi}[thm]{Definition}

\newtheorem{question}[thm]{Question}

\newtheorem{example}[thm]{Example}

\newtheorem{remark}[thm]{Remark}
\newenvironment{rmk}{\begin{remark} \em}{\end{remark}}

\begin{document}
\title{Symmetric uncoded caching schemes with low\\
 subpacketization levels}
\author{Tai Do Duc, Shuo Shao, Chaoping Xing}
\date{}

\maketitle

\begin{abstract}
Caching is a commonly used technique in content-delivery networks which aims to deliver information from hosting servers to users in the most efficient way.
In $2014$, Maddah-Ali and Niessen \cite{ali1} formulated caching into a formal information theoretic problem and it has gained a lot of attention since then.
It is known that the caching schemes proposed in \cite{ali1} and \cite{yu} are optimal, that is, they require the least number of transmissions from the server to satisfy all users' demands.
However for these schemes to work, each file needs to be partitioned into $F^*$ subfiles ($F^*$ is called the subpacketization level of files) with $F^*$ growing exponentially in the number $K$ of users.
As a result, it is problematic to apply these schemes in practical situations, where $K$ tends to be very large.
There rise the following questions: (1) are there optimal schemes in which each file is partitioned into $F$ subfiles, where $F$ is not exponential, say polynomial for example, in $K$? (2) if the answer to this question is no, is there a near-optimal scheme, a scheme which is as asymptotically good as the one in \cite{ali1,yu}, with $F$ polynomial in $K$? Both these questions are open.

Our main contribution in this paper is to provide answers to above questions. Firstly, we prove that under some mild restriction on user's cache rate, there are no optimal schemes with $F$ smaller than $F^*$. Moreover, we give necessary and sufficient conditions for the existence of optimal schemes in this case. Secondly, we provide an affirmative answer to  the second question raised above by an explicit construction and a detailed performance analysis.
\end{abstract}

\section{Introduction}
Caching is a common strategy used in data management in order to reduce network traffic congestion in peak times. This technique was studied since as early as $1982$ by Dowdy and Foster \cite{dow}.
In the caching setting, there is a placement phase and a deliver phase which are performed during off-peak times and peak times, respectively. In the placement phase, each user stores some data from the database in its cache.
These pre-stored data allow the server to reduce the amount of information distributed over the network during peak times (delivery phase).
At the early stage of research on caching \cite{alm,bae,bor,dow}, the gain by the server (or the reduction in the amount of information sent) merely comes from local duplication of the files in users' caches. This gain becomes negligible if the cache sizes are small compared to the amount of content stored in the server.
There is a need for a more systematic method to study the problem.

\medskip

In $2014$, Maddah-Ali and Niessen \cite{ali1} formulated caching into a formal information theoretic problem which has gained considerable attention from researchers in information theory.
Assume that there is a network consisting of one server with a database of $N$ files and there are $K$ users which are connected to the server through an error-free shared link.
Each user has a cache memory big enough to store $M$ of the files, where $M\leq N$ is a non-negative integer.

\begin{center}
\includegraphics[scale=0.8]{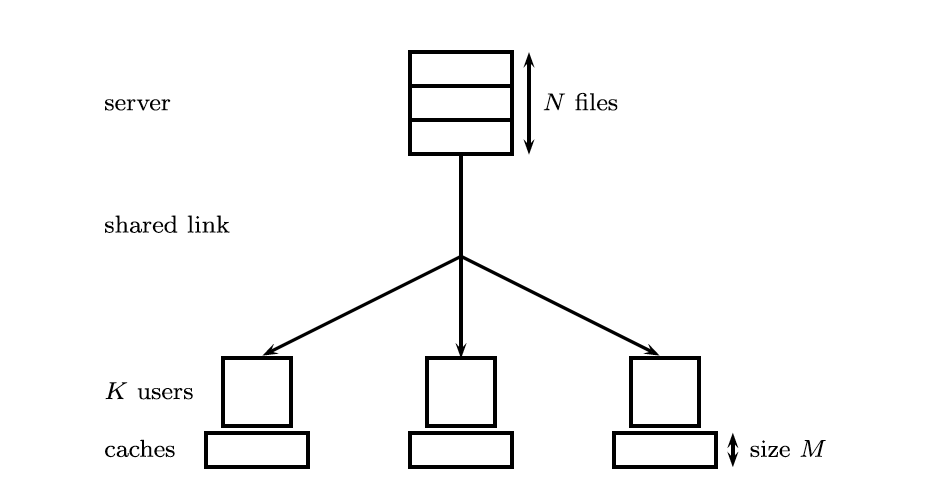}
\end{center}

\noindent A caching scheme consists of two phases, \textit{placement phase} and \textit{delivery phase}. 
\begin{enumerate}
\item Placement phase
\begin{itemize}
\item Each file is partitioned into $F$ subfiles, where $F$ is a positive integer. The number $F$ is called \textbf{subpacketization level} of the scheme. Note that there are totally $NF$ subfiles stored in the database.
\item Each user stores $MF$ linear combinations of these $NF$ subfiles in its cache.
\end{itemize}
\item Delivery phase \\
Each user requests one file and the server sends information to the users through the shared link to satisfy all users' requests.
\end{enumerate}
We define the \textbf{rate} $R$ of a caching scheme to be the smallest positive real number such that any users' demand can be met by $RF$ transmissions from the server. Given the values of $K,M,N$, the caching problem reduces to finding suitable values of $F$ so that we can design a scheme with the rate $R$ as smallest as possible. 

\medskip

Since the formal definition of the caching problem by Maddah-Ali and Niessen \cite{ali1}, there is an increasing interest in this line of research \cite{bha,chi,nie,pia,ram,shang,sha,tia,yan1,yan2,yu,yu2}. Among these works, the schemes proposed in \cite{ali1} and \cite{yu} are known to be optimal schemes. 

\subsection{Known optimal schemes and open questions} \label{optimalschemes}

It took a while, from $2014$ to $2016$, for Piantanida et. al. \cite{pia} to prove that if $N\geq K$, then the scheme in \cite{ali1} is optimal under the so-called \textit{uncoded} caching requirement. In an uncoded caching scheme, each user caches directly $MF$ subfiles from the $NF$ subfiles in the database. In a \textit{coded} caching scheme, each user is allowed stored $MF$ linear combinations of those $NF$ subfiles. We will only focus on uncoded caching schemes in this paper. 

\medskip

Continuing further on the work by Maddah-Ali and Niesen, Yu et. al. \cite{yu}, in $2018$, constructed an uncoded caching scheme which is optimal for any values of $K$ and $N$. This scheme has rate
\begin{equation}\label{rate_optimal}
R^*=\frac{K-KM/N}{1+KM/N}-\frac{{K-\min\{K,N\}\choose KM/N+1}}{{K\choose KM/N}}.
\end{equation}
We have used the term \textit{optimal} many times but have not justified it clearly until this point. From now on, we call an uncoded caching scheme with $K$ users, $N$ files and cache size $M$ optimal if it has rate $R=R^*$, where $R^*$ is defined by (\ref{rate_optimal}). Despite new constructions of numerous caching schemes (see \cite{bha,chi,ram,shang,yan1,yan2} for examples) since Maddah-Ali and Niesen's formal formulation of the caching problem, the scheme proposed by Yu et. al. remains the only known uncoded caching scheme which attains the rate $R^*$. 

\medskip

Nevertheless, there is an unpleasant problem which arises from both schemes proposed in \cite{ali1} and in \cite{yu}. For these schemes to work, each file needs to be partitioned into 
\begin{equation} \label{subpacket_optimal}
F^*={K\choose KM/N}
\end{equation}
subfiles.
As $F^*$ grows exponentially in $K$, these schemes may not be relevant for many practical implementations which require $K$ to be large.
There has been considerable effort, starting from $2016$ by Shanmugam at. el. \cite{sha}, to remedy this problem by constructing new schemes which have subpacketization level $F$ smaller than $F^*$ while not increasing the value of $R^*$ by too much, see \cite{bha,chi,ram,shang,yan1,yan2}.
On the other hand, an obvious better solution to this problem is to find an optimal scheme with subpacketization level $F$ smaller than $F^*$, or even better, $F$ polynomial in $K$.
Though there is clear suggestion on a trade-off between $R$ and $F$, that is, $R$ is small if $F$ is large and vice versa, an understanding on this trade-off remains vague.
We summarize our discussion in this paragraph into the following questions.

\medskip

\begin{question} \label{question1}
Let $K$ and $N$ be fixed positive integers. Let $M$ be a nonnegative integer such that $M\leq N$. Define $R^*$ and $F^*$ as in (\ref{rate_optimal}) and (\ref{subpacket_optimal}), respectively. Is there an uncoded caching scheme with $K$ users, $N$ files and cache size $M$ whose rate is $R=R^*$ and whose subpacketization level $F$ is smaller than $F^*$? Furthermore if it is possible, classify the subpacketization levels of optimal uncoded caching schemes.
\end{question}

\begin{question} \label{question2}
If the answer to Question $1$ is no, is there an uncoded caching scheme with rate $R$ asymptotically close to $R^*$ and subpacketization level $F$ polynomial in the number $K$ of users?
\end{question}

\subsection{Best known results and our results}
We remark that both Question \ref{question1} and Question \ref{question2} are open and Question \ref{question2} was asked by the authors in \cite{ram} and \cite{shang}.
The best known results toward Question \ref{question2} are the ones by Yan et. al. \cite{yan1} in $2017$ and Shangguan et. al. \cite{shang} in $2018$. In \cite{yan1}, the authors constructed two uncoded caching schemes with $R$ asymptotically close to $R^*$ and $F$ smaller than $F^*$ by a factor which is exponential in $K$.
However in these schemes, we still have $F$ exponentially large compared to $K$.
In \cite{shang}, the authors made a further improvement by proposing two new schemes which have $F$ sub-exponential in $K$.
\medskip

On the other hand, there is not much progress on Question \ref{question1}, as no optimal scheme with $F<F^*$ has been found.
The best known work on Question \ref{question1} is done in \cite{yan1} where the authors showed that $F^*$ is the smallest subpacketization level of an optimal scheme.
However, this result is only applied to a special class of uncoded caching schemes, called Placement Delivery Arrays (PDAs), and it does not provide us an insight on sufficient conditions for the existence of optimal uncoded caching schemes.
In summary, do the scarcity of optimal schemes with $F<F^*$ and the result on PDAs suggest that these schemes simply do not exist?

\medskip

Our main contribution in this paper is to provide answers to both Question \ref{question1} and Question \ref{question2}.
Our results are obtained under the restrictions of a \textit{symmetric} uncoded caching scheme, a natural property which is present in all currently known uncoded caching schemes (see Definition \ref{symmetric} for definition). In these schemes, each user caches the same fraction of each file and each subfile is cached by the same number of users. Our first main result is the following.

\begin{thm} \label{main_lowerbound0}
Let $K$ and $N$ be positive integers. Let $M$ be a non-negative integer such that $M\leq N$. In any symmetric uncoded caching scheme with $K$ users, $N$ files, cache size $M$, and rate $R$, we have
$$R\geq R^*.$$
Furthermore if 
$\frac{M}{N}\leq \frac{\min\{K,N\}}{K},$
then there exists a symmetric uncoded caching scheme with rate $R=R^*$
if and only if $F\equiv 0\pmod{F^*}$.
\end{thm}

\noindent For Theorem \ref{main_lowerbound0} to work, the user's cache rate $\frac{M}{N}$ need to not exceed $\frac{\min\{K,N\}}{K}$.
Under this condition, we prove that the necessary and sufficient condition for the existence of an optimal scheme is $F\equiv 0\pmod{F^*}$, which implies the non-existence of an optimal scheme with $F<F^*$.
We note that if $N\geq K$, then the inequality $\frac{M}{N}\leq 1=\frac{\min\{K,N\}}{K}$ holds automatically and our result confirms that in this case, there is no optimal scheme with $F<F^*$.
Nevertheless, there is a minor open case in Theorem \ref{main_lowerbound0}, which is the case $K>N$ and $\frac{M}{N}>\frac{N}{K}$.
Despite this open case, our result strongly hints towards the suggestion that there is no symmetric uncoded caching scheme with rate $R=R^*$ and subpacketization level $F<F^*$.
In order to find an uncoded caching scheme with $R=R^*$ and $F<F^*$, one may need look to the direction of non-symmetric schemes, which is still a completely open land.

\medskip

Our second main result of this paper is the following.
\begin{thm} \label{main_constr0}
Let $n$ be a positive integer. Let $a$ and $b$ be non-negative integers such that $a+b\leq n$. Then there exists a symmetric uncoded caching scheme with the following parameters.
$$K={n\choose a}, \ F={n\choose b}, \ \frac{M}{N}=\frac{{n\choose b}-{n-a\choose b}}{{n\choose b}}, \ R=\frac{{n\choose a+b}}{{n\choose b}}.$$
Furthermore, let $\epsilon> 0$ be a positive real number. The above scheme, with suitable choices of $a,b,n$, has parameters $R,K$ and $F$ satisfying the following conditions.
\begin{itemize}
\item[(i)] $R^*\leq R \leq R^*(1+\epsilon)$,
\item[(ii)] $K\leq F\leq K^{1+\epsilon}$, and
\item[(iii)] $F^*\geq F^{(\log F)^{1/\epsilon}}$.
\end{itemize}
\end{thm}

\noindent The conditions (i)-(iii) clearly imply that the scheme in Theorem \ref{main_constr0} has rate $R$ asymptotically close to $R^*$ and subpacketization level $F$ polynomial in $K$.
Moreover $F$ is sub-exponentially smaller than $F^*$.
Thus, Theorem \ref{main_constr0} settles Question \ref{question2} completely.
However, after discovering the scheme in Theorem \ref{main_constr0}, we noticed, in the process of literature review for this paper, that our scheme was already known by Shangguan et. al. \cite{shang} in $2018$  via the language of hypergraph.
The reason that Question \ref{question2} is still open lies in the complexity of analyzing this scheme, as the authors in \cite{shang} also commented.
While our merit for solving Question \ref{question2} is a detailed analysis on the performance of the mentioned scheme, the construction of the scheme is fully credited to Shangguan et. al. \cite{shang}.

\medskip
\subsection{Organization}
The remaining of this paper is organized as follows. In Section \ref{preli}, we provide some definitions and technical lemmas which will be used throughout the paper.
In Section \ref{section_lowerbound} and Section \ref{section_construction}, we prove Theorem \ref{main_lowerbound0} and Theorem \ref{main_constr0}. In Section \ref{section_conclusion}, we conclude the paper with several open questions in this research direction.

\section{Preliminaries} \label{preli}
In this section, we state definitions and introduce notations which will be used throughout the rest of the paper.
Let $k$ and $n$ be positive integers such that $k\leq n$. We use the following notations.
\begin{itemize}
\item We denote the set $\{1,\dots,n\}$ by $[n]$. 
\item We call $A$ a $k$-subset of $[n]$ if $|A|=k$.
\item We call $A$ an ordered $k$-subset of $[n]$ if there are $k$ distinct elements $a_1,\dots,a_k\in [n]$ such that $A=(a_1,\dots,a_k)$.
\end{itemize}

\medskip

Next, we define uncoded caching schemes.
\begin{defi} \label{scheme}
Let $K,N$ and $F$ be positive integers. Let $M$ be a non-negative integer such that $M\leq N$.
We call a caching scheme \textbf{uncoded caching scheme} with parameters $K,M,N,F,R$ if it has the following properties.
\begin{itemize}
\item[(i)] There are $K$ users and $N$ files. \\
Denote the files by $W_1,\dots, W_N$. Each file $W_i$ is partitioned into $F$ subfiles $W_{i,1}, \dots, W_{i,F}$. We call $F$ the subpacketization level of the scheme.
\item[(ii)] In the placement phase, each user is allowed to store $MF$ subfiles $W_{i,j}$ in its cache. We call $M/N$ the \textbf{user's cache rate}, that is, each user caches on average $MF/N$ subfiles from a file. 
\item[(iii)] In the delivery phase
\begin{itemize}
\item Each user requests one file and the server sends transmissions to the users, each as a linear combination of $W_{i,j}$'s, to satisfy all users' requests.
\item The number $R$ called the rate of the scheme. It is the smallest positive real number such that any demand of the users can be met by $RF$ transmissions from the server.
\end{itemize}
\end{itemize}
\end{defi}

\medskip

By our knowledge, all existing uncoded caching schemes are symmetric in the following sense. 

\begin{defi}\label{symmetric}
We call an uncoded caching scheme \textbf{symmetric} if
\begin{itemize}
\item[(i)] Each user caches the same fraction of each file. That is, if a user caches $W_{i,j}$, then he or she also caches $W_{k,j}$ for any $k=1,\dots,N$.
\item[(ii)] Each subfile $W_{i,j}$ is cached by the same number of users.
\end{itemize}
\end{defi}

\noindent We will only focus on symmetric uncoded caching schemes in this paper.
We conclude this section with a simple observation on symmetric uncoded caching schemes which will be used repeatedly in Section \ref{section_lowerbound} and Section \ref{section_construction}.

\begin{lem} \label{observation}
In a symmetric uncoded caching scheme with parameters $K,M,N,R,F$, we have the following.
\begin{itemize}
\item[(a)] For each fixed $i\in [N]$, each user stores exactly $Z=MF/N$ subfiles $W_{i,j}$ from the file $W_i$.
\item[(b)] For each $i\in [N]$ and $j\in [F]$, the subfile $W_{i,j}$ is in the caches of exactly $t=KM/N$ users.
\end{itemize}
\end{lem}

\begin{proof}
\item[(a)] Fix an user $U$. By Definition \ref{symmetric}.(i), there exists a positive integer $Z$ such that $U$ stores exactly $Z$ subfiles of any file in its cache. Hence the cache of $U$ has $NZ$ subfiles. We obtain $NZ=MF$ by Definition \ref{scheme}.(ii), which implies $Z=MF/N$.
\item[(b)] By Definition \ref{symmetric}.(ii), there exists a positive integer $t$ such that each subfile is cached by exactly $t$ users. Counting the number of pairs $(U,S)$, where $U$ is an user and $S$ is a subfile contained in the cache of $U$, in two ways, we obtain $KMF=tNF$, which implies $t=KM/N$.
\end{proof}

\bigskip

\section{Subpacketization levels of optimal schemes}\label{section_lowerbound}
In this section, we prove Theorem \ref{main_lowerbound0}. First, we recall the theorem for the convenience of the readers.
\begin{thm} \label{main_lowerbound}
Let $K$ and $N$ be positive integers. Let $M$ be a non-negative integer such that $M\leq N$. Consider any symmetric uncoded caching scheme with $K$ users, $N$ files, user's cache rate $\frac{M}{N}$ and rate $R$. We have
\begin{equation}\label{optimal}
R\geq \frac{K-KM/N}{1+KM/N}-\frac{{K-\min\{K,N\}\choose KM/N+1}}{{K\choose KM/N}}.
\end{equation}
Furthermore if 
$\frac{M}{N}\leq \frac{\min\{K,N\}}{K},$
then there exists a symmetric uncoded caching scheme with rate 
$R=\frac{K-KM/N}{1+KM/N}-\frac{{K-\min\{K,N\}\choose KM/N+1}}{{K\choose KM/N}}$
if and only if $F\equiv 0\pmod{{K\choose KM/N}}$.
\end{thm}

\noindent We remark that the authors in \cite{yu} proved (\ref{optimal}) by showing that
$$R_F\geq \frac{K-KM/N}{1+KM/N}-\frac{{K-\min\{K,N\}\choose KM/N+1}}{{K\choose KM/N}}-\frac{c_0}{F},$$
where $R_F$ denotes the rate of a symmetric uncoded caching scheme with subpacketization level $F$ and $c_0$ is a constant independent of $F$. By letting $F$ tend to infinity, they obtain (\ref{optimal}).

\medskip

In this section, we give another proof for (\ref{optimal}) from which we can draw a conclusion on $F$ in the case of equality in (\ref{optimal}), or in other words, classify the subpacketization level of an optimal scheme.
Our proof for Theorem \ref{main_lowerbound} is divided into two lemmas. In the first lemma, we prove (\ref{optimal}). In the second lemma, we classify $F$ in the case of equality in (\ref{optimal}).

\begin{lem} \label{main1}
 In a symmetric uncoded caching scheme with $K$ users, $N$ files, user's cache rate $\frac{M}{N}$ and rate $R$, we have
$$R\geq \frac{K-KM/N}{1+KM/N}-\frac{{K-\min\{K,N\}\choose KM/N+1}}{{K\choose KM/N}}.$$
\end{lem}

\begin{proof} 
The proof is divided into two cases, $K\geq N$ and $K<N$.
\medskip

\textbf{Case 1.} Assume $K\geq N$.
In this case, we need to prove that
\begin{equation}\label{opt1}
R\geq \frac{K-KM/N}{1+KM/N}-\frac{{K-N\choose KM/N+1}}{{K\choose KM/N}}.
\end{equation}
Let $W_1,\dots,W_N$ denote $N$ files and let $1,\dots,K$ denote $K$ users of the scheme. Let $D=(d_1,\dots,d_N)$ be any ordered $N$-subset of $[K]$. Assume that user $d_i$ requests file $W_i$, $i=1,\dots, N$. Consider a virtual user $V_D$ whose cache is filled as follows.
\begin{itemize}
\item Step $1$. Add all subfiles of $W_{{i}}, \ i=1,\dots,N,$ which are in the cache of user $d_1$ to $V_D$. 
\item Step $2$. Add all subfiles of $W_{{i}}, \ i=2,\dots,N,$ which are in the cache of user $d_2$ to $V_D$. \\
$\dots\dots\dots\dots$
\item Step $k$. Add all subfiles of $W_{{i}}, \ i=k,\dots,N,$ which are in the cache of user $d_k$ to $V_D$. \\
$\dots\dots\dots\dots$
\item Step $N$. All all subfiles of $W_{N}$ which are in the cache of user $d_N$ to $V_D$.
\end{itemize}

\noindent After receiving $RF$ transmissions from the server, user $V_D$ can proceed inductively to decode $W_{{1}},\dots,W_{{N}}$. Next, we look at the cache size of $V_D$.

\medskip

Assume that each file $W_i$ is partitioned into $F$ subfiles $W_{i,j}$, $j=1,\dots,F$. Let $U_i$ denote the set of indices $j$ such that the subfiles $W_{1,j}$ are in the cache of user $i$, that is,
$$U_i=\{j\in [F]: \ W_{1,j} \ \text{is in the cache of user} \ i\}.$$
Note that $|U_i|=Z=\frac{MF}{N}$ for any $i\in [K]$ by Lemma \ref{observation}. Moreover due to the symmetry of the scheme, the subfiles which are in the cache of user $i$, $i\in [K]$, are $\{W_{l,j}: l\in [N], j\in U_i\}$. The cache of $V_D$ includes the following subfiles.
\begin{itemize}
\item $|U_{d_1}|$ subfiles of $W_1$.
\item $|U_{d_1}\cup U_{d_2}|$ subfiles of $W_2$. 
\item $|U_{d_1}\cup U_{d_2}\cup U_{d_3}|$ subfiles of $W_3$. \\
$\dots\dots\dots\dots$
\item $|U_{d_1}\cup\cdots\cup U_{d_N}|$ subfiles of $W_N$.
\end{itemize} 
The number of subfiles in the cache of $V_D$ is $C_D=\sum_{k=1}^N |U_{d_1}\cup\cdots\cup U_{d_k}|$.
As $V_{D}$ is able to decode all $NF$ subfiles of $W_{1},\dots, W_{N}$, the server needs to send at least $NF-C_{D}$ transmissions. Hence
\begin{equation} \label{intersec}
RF\geq NF-\sum_{k=1}^N |U_{d_1}\cup\cdots\cup U_{d_k}|.
\end{equation}
Taking (\ref{intersec}) over all ordered $N$-subsets $D=(d_1,\dots,d_N)$ of $[K]$, we obtain
\begin{eqnarray*}
\frac{K!}{(K-N)!} RF &\geq & \frac{K!}{(K-N)!}NF-\sum_{(d_1,\dots,d_N)}\sum_{k=1}^N |U_{d_1}\cup\cdots\cup U_{d_k}| \\
&=& \frac{K!}{(K-N)!}NF-\sum_{k=1}^N \sum_{(d_1,\dots,d_N)} |U_{d_1}\cup\cdots\cup U_{d_k}| 
\end{eqnarray*} 
Note that each term $|U_{d_1}\cup\cdots\cup U_{d_k}|$ appears exactly $\frac{k!(K-k)!}{(K-N)!}$ times in the sum above. So
\begin{equation}\label{lowerbound}
\frac{K!}{(K-N)!}RF \geq \frac{K!}{(K-N)!}NF-\sum_{k=1}^N\frac{k!(K-k)!}{(K-N)!}\sum_{\substack{\{d_1,\dots, d_k\}\subset [K]\\ |\{d_1,\dots,d_k\}|=k}} |U_{d_1}\cup\cdots\cup U_{d_k}|.
\end{equation}
Note that for each $j\in[F]$, there are $t=\frac{KM}{N}$ sets $U_i$ which contain $j$ (see Lemma \ref{observation}).
By counting the number of $(k+1)$-sets $\{d_1,\dots,d_k,j\}$ in which $d_i\in [K]$ for all $i$ and $j\in [F]$ such that $j\in \left( U_{d_1}\cup\cdots\cup U_{d_k}\right)$, we obtain 
\begin{equation} \label{count}
\sum_{\substack{\{d_1,\dots, d_k\}\subset [K]\\ |\{d_1,\dots,d_k\}|=k}} |U_{d_1}\cup\cdots\cup U_{d_k}|=F\left({K\choose k}-{K-t\choose k}\right).
\end{equation}
By (\ref{lowerbound}) and (\ref{count}), we obtain
\begin{eqnarray*}
\frac{K!}{(K-N)!}RF&\geq & \frac{K!}{(K-N)!}NF-\sum_{k=1}^N\frac{k!(K-k)!}{(K-N)!}F\left({K\choose k}-{K-t\choose k}\right) \\
&=& \frac{K!}{(K-N)!}NF-\frac{K!}{(K-N)!}NF+\frac{F}{(K-N)!}\sum_{k=1}^{N}\frac{(K-t)!(K-k)!}{(K-k-t)!},
\end{eqnarray*}
which implies
\begin{eqnarray*}
R &\geq & \sum_{k=1}^{N}\frac{(K-t)!(K-k)!}{K!(K-k-t)!} =\frac{\sum_{k=1}^N {K-k\choose t}}{{K\choose t}} = \frac{{K\choose t+1}-{K-N\choose t+1}}{{K\choose t}},
\end{eqnarray*}
where in the last equality, we use
$$\sum_{k=m+1}^n{k\choose l}=\sum_{k=1}^n {k\choose l}-\sum_{k=1}^m{k\choose l}={n+1\choose l+1}-{m+1\choose l+1}$$
for any positive integers $l,m,n$ with $m< n$. Continuing on the last inequality on $R$ and noting that $t=KM/N$, we obtain
$$R\geq \frac{K-KM/N}{1+KM/N}-\frac{{K-N\choose KM/N+1}}{{K\choose KM/N}},$$ 
proving (\ref{opt1}). 

\medskip

\textbf{Case 2.} Assume $K<N$. In this case, we need to show that 

\begin{equation}\label{opt2}
R\geq \frac{K-KM/N}{1+KM/N}.
\end{equation}
The idea for the proof of this case is similar to that of the last case, but with a little switch.
Let $D=(d_1,\dots,d_K)$ be any permutation of the set $[K]$.
Assume that user $d_i$ requests file $W_i$, $i=1,\dots,K$.
We also consider a virtual user $V_D$ whose cache is filled as follows.
\begin{itemize}
\item \textbf{Step $1$.} Add all subfiles of $W_i, \ i=1,\dots,N,$ which are in the cache of user $d_1$ to $V_D$. 
\item \textbf{Step $2$.} Add all subfiles of $W_i, \ i=2,\dots,N,$ which are in the cache of user $d_2$ to $V_D$. \\
$\dots\dots\dots\dots$
\item \textbf{Step $k$.} Add all subfiles of $W_i, \ i=k,\dots,N,$ which are in the cache of user $d_k$ to $V_D$. \\
$\dots\dots\dots\dots$
\item \textbf{Step $K$.} All all subfiles of $W_i, \ i=K,\dots,N,$ which are in the cache of user $d_K$ to $V_D$.
\end{itemize}
After receiving $RF$ transmissions from the server, $V_D$ can proceed inductively to decode $W_1,\dots,W_K$.
Define the sets $U_i, \ i=1,\dots, K$, as in the last case. The cache of user $V_D$ contains the following.

\begin{itemize}
\item $|U_{d_1}|$ subfiles of $W_1$.
\item $|U_{d_1}\cup U_{d_2}|$ subfiles of $W_2$. \\
$\dots\dots\dots\dots$
\item $|U_{d_1}\cup\cdots\cup U_{d_K}|$ subfiles of $W_K$.
\item $|U_{d_1}\cup\cdots\cup U_{d_K}|$ subfiles of $W_{K+1}$.\\
$\dots\dots\dots\dots$
\item $|U_{d_1}\cup\cdots\cup U_{d_K}|$ subfiles of $W_N$.
\end{itemize} 
The additional switch we mentioned is the following.
For each file $W_i$, $i=K+1,\dots,N$, we send the missing $F-|U_{d_1}\cup\cdots\cup U_{d_K}|$ subfiles of $W_i$ to $V_D$. After receiving $RF$ transmissions from the server and the extra $(N-K)(F-|U_{d_1}\cup\cdots\cup U_{d_K}|)$ missing subfiles, $V_D$ can decode all $N$ files $W_1,\dots,W_N$.
By similar reasoning as the last case, we have
$$RF+(N-K)(F-|U_{d_1}\cup\cdots\cup U_{d_K}|)\geq NF-\left(\sum_{k=1}^K|U_{d_1}\cup\cdots\cup U_{d_k}|+(N-K)|U_{d_1}\cup\cdots\cup U_{d_K}|\right),$$
which implies 
\begin{equation}\label{intersec2}
RF\geq KF-\sum_{k=1}^K|U_{d_1}\cup\cdots\cup U_{d_k}|.
\end{equation}
Taking (\ref{intersec2}) over all permutations $D=(d_1,\dots,d_K)$ of $[K]$, we obtain
\begin{eqnarray*}
K!RF &\geq & K!KF-\sum_{k=1}^K\sum_{(d_1,\dots,d_K)}|U_{d_1}\cup\cdots\cup U_{d_k}| \\
&=& K!KF-\sum_{k=1}^K k!(K-k)!\sum_{\substack{\{d_1,\dots,d_k\}\subset [K]\\ |\{d_1,\dots,d_k\}|=k}} |U_{d_1}\cup\cdots\cup U_{d_k}| \\
&=& K!KF-\sum_{k=1}^K k!(K-k)! \left({K\choose k}-{K-t\choose k}\right)F,
\end{eqnarray*}
where the last equality follows from (\ref{count}). Continuing on the last inequality on $R$, we obtain
$$R\geq \sum_{k=1}^K\frac{(K-t)!(K-k)!}{K!(K-k-t)!}=\sum_{k=1}^K\frac{{K-k\choose t}}{{K\choose t}}=\frac{{K\choose t+1}}{{K\choose t}}=\frac{K-t}{1+t},$$
proving (\ref{opt2}).
\end{proof}

\medskip

In the next lemma, we classify the case of equality in (\ref{optimal}) to complete the proof of Theorem \ref{main_lowerbound}.

\begin{lem} \label{main2}
There exists a symmetric uncoded caching scheme with $K$ users, $N$ files, user's cache rate $\frac{M}{N}$ and rate $R$ satisfying
\begin{equation}\label{conditions}
R=\frac{K-KM/N}{1+KM/N}-\frac{{K-\min\{K,N\}\choose KM/N+1}}{{K\choose KM/N}} \ \ \text{and} \ \ \frac{M}{N}\leq \frac{\min\{N,K\}}{K}
\end{equation}
if and only if
\begin{equation} \label{sufficient}
F \equiv 0\pmod{{K\choose KM/N}}.
\end{equation}
\end{lem}
\begin{proof}
First, we consider the case $K\geq N$. Recall that for any $i\in[K]$, we define
$$U_i=\{1\leq j\leq F: \ W_{1,j} \ \text{is in the cache of user} \ i\}.$$
By the proof of Lemma \ref{main1}, the equality 
$$R=\frac{K-KM/N}{1+KM/N}-\frac{{K-N \choose KM/N+1}}{{K\choose KM/N}}$$
implies that all inequalities (\ref{intersec}) become equalities, that is, all sums
\begin{equation}\label{sums}
S_{(d_1,\dots,d_N)}=|U_{d_1}|+|U_{d_1}\cup U_{d_2}|+\cdots+|U_{d_1}\cup\cdots\cup U_{d_N}|
\end{equation}
are the same over all ordered $N$-subsets $(d_1,\dots,d_N)$ of $[K]$.

\textbf{Claim.} For any fixed $k\in [N]$, the terms $T_{I_k}=|\cup_{i\in I_k} U_i|$
are the same over all choices of $k$-subsets $I_k$ of $[K]$.\\
\textit{Proof of Claim.} If $k=1$, then it is clear that the claims holds because $|U_i|=Z$ for any $i$. From now on, we assume $k\geq 2$. Let $I_k$ and $J_k$ be any two $k$-subsets of $[K]$. We prove $T_{I_k}=T_{J_k}$ by induction on the intersection size $|I_k\cap J_k|$.

If $|I_k\cap J_k|=k$, then $I_k=J_k$ and it is clear that $T_{I_k}=T_{J_k}$.
Next, assume that $|I_k\cap J_k|=k-1$. Write $I_k=\{i_1,\dots,i_{k-1},i\}$ and $J_k=\{i_1,\dots,i_{k-1},j\}$. 
If $k=N$, then using $S_{(i_1,\dots,i_{k-1},i)}=S_{(i_1,\dots,i_{k-1},j)}$ from (\ref{sums}), we obtain $T_{I_k}=T_{J_k}$.
Assume $k<N$.
Let $\{d_{k+2},\dots,d_N\}$ be any subset $[K]$ which has empty intersection with $I_k\cup J_k$.
This set is empty if $N=k+1$.
Using
$$S_{(i_1,\dots,i_{k-1},i,j,d_{k+2},\dots,d_N)}=S_{(i_1,\dots,i_{k-1},j,i,d_{k+2},\dots,d_N)},$$
from (\ref{sums}), we obtain $T_{I_k}=T_{J_k}$.
Thus $T_{I_k}=T_{J_k}$ in the case $|I_k\cap J_k|=k-1$.

Assume $T_{I_k}=T_{J_k}$ for $|I_k\cap J_k|\in \{l,l+1,\dots,k\}$, where $l\leq k-1$ is a positive integer.
Now, suppose that $I_k$ and $J_k$ are any two $k$-subsets of $[K]$ such that $|I_k\cap J_k|=l-1$.
Write $$I_k=\{c_1,\dots,c_{l-1},i_l,\dots,i_k\}, \ J_k=\{c_1,\dots,c_{l-1},j_l,\dots,j_k\}.$$
Define $I=\{c_1,\dots,c_{l-1},i_l,\dots,i_{k-1},j_k\}$. Note that $|I\cap I_k|=k-1\geq l$ and $|I\cap J_k|=l$.
By the inductive assumption, we obtain
$$T_{I_k}=T_I=T_{J_k},$$
proving the claim.

\medskip

Now, we use the claim to finish the proof for the case $K\geq N$. By the claim, it is clear (by induction on $k$) that for any $k\in [N]$, all intersections $|U_{d_1}\cap \cdots\cap U_{d_k}|$ are the same over all choices of $k$-subsets $\{d_1,\dots,d_k\}$ of $[K]$. Next, note that $t=KM/N\leq N$, as $M/N\leq N/K$ by (\ref{conditions}). Fix a positive integer $k\leq t$. Counting the number of $(k+1)$-sets $\{d_1,\dots,d_k,j\}$ in which $d_i\in [K]$ for all $i$ and $j\in [F]$ such that $j\in \left( U_{d_1}\cap\cdots\cap U_{d_k} \right)$, we obtain

\begin{equation}\label{keyequa}
\sum_{\substack{\{d_1,\dots,d_k\}\subset [K] \\ |\{d_1,\dots,d_k\}|=k}} |U_{d_1}\cap \cdots \cap U_{d_k}|={t \choose k}F.
\end{equation}
In (\ref{keyequa}), letting $k=t$ and noting that all terms $|U_{d_1}\cap\cdots\cap U_{d_t}|$ have the same value (this holds because $t\leq N$), we obtain
$$F\equiv 0 \pmod{{K\choose t}},$$
proving (\ref{sufficient}) in the case $K\geq N$. 

\medskip

The case $K<N$ is proved in the exact same way as the last case.
In this case, all inequalities (\ref{intersec2}) become equalities. Thus all sums  
$$S_{(d_1,\dots,d_K)}=\sum_{k=1}^K|U_{d_1}\cup\cdots\cup U_{d_k}|$$
are the same over all choices of permutations $(d_1,\dots,d_K)$ of $[K]$. Using this property, we obtain that for each $k\in [K]$, all terms $T_{I_k}=|\cup_{i\in I_k}U_i|$ are the same over all choices of $k$-subsets $I_k$ of $[K]$.
Hence the intersections $|U_{d_1}\cap \cdots \cap U_{d_k}|$ are the same over all choices of $k$-subsets $\{d_1,\dots,d_k\}$ of $[K]$.
We obtain equation (\ref{keyequa}) and the congruence $F\equiv 0 \pmod{{K\choose t}}$ is achieved by letting $k=t$ in this equation.
Note that we always have $t=KM/N\leq K$ in this case and it is safe to let $k=t$ in (\ref{keyequa}).
The details are left to the readers.

\medskip

Lastly, it remains to prove that if $F\equiv 0\pmod{{K\choose KM/N}}$, then there exists a symmetric uncoded caching scheme with rate $R= \frac{K-KM/N}{1+KM/N}-\frac{{K-\min\{K,N\}\choose KM/N+1}}{{K\choose KM/N}}$. In fact, the proposed scheme does not require $\frac{M}{N}\leq \frac{\min\{N,K\}}{K}$. This scheme is a slight modification of the scheme in \cite[Section III.B]{yu} and is presented in the appendix.

\medskip

\end{proof}

\medskip

\begin{rmk} \label{remark}
In Theorem \ref{main_lowerbound}, we prove that there is no symmetric uncoded caching scheme with rate $R=R^*$ and subpacketization level $F<F^*$ if the user's cache rate $\frac{M}{N}$ does not exceed $\frac{\min\{K,N\}}{K}$. The remaining open case is $K>N$ and $\frac{M}{N}>\frac{N}{K}$. In this case, we have $t=KM/N>N$ and the equation (\ref{keyequa}) implies 
\begin{equation}\label{congr}
\sum_{\substack{\{d_1,\dots,d_k\}\subset \{1,\dots,K\} \\ |\{d_1,\dots,d_k\}|=k}} |U_{d_1}\cap \cdots \cap U_{d_k}|={t\choose k}F \ \ \text{for any} \ \ k=1,\dots,N.
\end{equation}
Note that for each $k\leq N$, all terms on the left-hand side of (\ref{congr}) are the same, which implies
$$F\equiv 0 \ \mod \ \frac{{K\choose k}}{\gcd\left({K \choose k},{t\choose k}\right)} \ \text{for} \ k=1,\dots,N.$$
An open question is whether these congruence equations imply either $F\equiv 0\pmod{F^*}$ or $F>F^*$.
\end{rmk}

\bigskip

\section{A near-optimal scheme} \label{section_construction}
In this section, we give a detailed analysis on the performance of the scheme \cite[Construction I]{shang} in order to provide an affirmative answer to Question \ref{question2} proposed in the introduction.
The authors in \cite{shang} proposed this scheme via the language of hypergraph. 
We will not use this graph theoretic approach in our study. 
To make our result self-contained, we include both a description of the scheme and a simple proof for its implementability.
We recall Theorem \ref{main_constr0} for the convenience of the readers.

\begin{thm} \label{mainconstr}
Let $n$ be a positive integer. Let $a$ and $b$ be non-negative integers such that $a+b\leq n$. Then there exists a symmetric uncoded caching scheme with the following parameters.
\begin{equation}\label{para}
K={n\choose a}, \ F={n\choose b}, \ \frac{M}{N}=\frac{{n\choose b}-{n-a\choose b}}{{n\choose b}}, \ R=\frac{{n\choose a+b}}{{n\choose b}}.
\end{equation}
Furthermore, let $\epsilon> 0$ be a positive real number. The above scheme, with suitable choices of $a,b,n$, has parameters $R,K$ and $F$ satisfying the following conditions.
\begin{itemize}
\item[(i)] $R^*\leq R \leq R^*(1+\epsilon)$,
\item[(ii)] $K\leq F\leq K^{1+\epsilon}$, and
\item[(iii)] $F^*\geq F^{(\log F)^{1/\epsilon}}$.
\end{itemize}
\end{thm}

In preparation for the proof of Theorem \ref{mainconstr}, we prove the following lemma on the approximation of binomial coefficients which will be used repeatedly later.

\begin{lem} \label{approx_bin}
Let $f(n)$ and $g(n)$ be positive integers which are functions of $n$ such that 
$$\lim_{n\rightarrow\infty} \frac{f(n)}{g(n)}=\lim_{n\rightarrow\infty}\frac{f(n)^2}{g(n)}=0.$$ 
Then
\begin{equation} \label{approx}
\lim_{n\rightarrow \infty}{g(n)\choose f(n)}\times \frac{f(n)!}{g(n)^{f(n)}}=1.
\end{equation}
\end{lem}

\begin{proof}
We have
$$\frac{(g(n)-f(n))^{f(n)}}{f(n)!}\leq  {g(n)\choose f(n)}=\frac{g(n)(g(n)-1)\dots (g(n)-f(n)+1)}{f(n)!}\leq \frac{g(n)^{f(n)}}{f(n)!},$$
so
\begin{equation}\label{bound}
\left(1-\frac{f(n)}{g(n)}\right)^{f(n)} \leq {g(n)\choose f(n)}\frac{f(n)!}{g(n)^{f(n)}}\leq 1.
\end{equation}
Note that 
\begin{equation} \label{bound2}
\lim_{n\rightarrow\infty}\left(1-\frac{f(n)}{g(n)}\right)^{f(n)}=\lim_{n\rightarrow\infty} \left(1-\frac{f(n)}{g(n)}\right)^{\frac{g(n)}{f(n)} \frac{f(n)^2}{g(n)}}=\lim_{n\rightarrow\infty} e^{-\frac{f(n)^2}{g(n)}}=1.
\end{equation}
The equation (\ref{approx}) follows from (\ref{bound}) and (\ref{bound2}).
\end{proof}

\medskip

Now we are ready for the proof of Theorem \ref{mainconstr}. 

\begin{proof}[Proof of Theorem \ref{mainconstr}]
We define a scheme as follows.
\begin{enumerate}
\item Each user is labeled by a subset $A$ of $[n]$ such that $|A|=a$. The number of users is $K={n\choose a}$.
\item Assume that $N$ files are $W_1,\dots,W_N$. Each file $W_i$ is partitioned into $F={n \choose b}$ subfiles $\{W_{i,B}: B\subset [n], |B|=b\}$.
\item In the placement phase, user $A$ caches subfile $W_{i,B}$, $1\leq i\leq N$, if and only if $A\cap B\neq \emptyset$. In this way, user $A$ caches 
$$MF=N\left( {n \choose b}-{n-a \choose b}\right)$$
subfiles. The user cache rate $M/N$ is
$$\frac{M}{N}=\frac{{n \choose b}-{n-a \choose b}}{{n \choose b}}.$$
\item In the delivery phase, assume that user $A$ requests file $W_{d_A}$. For each subset $C$ of $[n]$ of size $|C|=a+b$, the server sends
$$Y_C=\sum_{A'\subset C: |A'|=a}W_{d_{A'},C\setminus A'}.$$
Note that the server needs to send $RF={n\choose a+b}$ messages, so
$R={n\choose a+b}{n \choose b}^{-1}$.
\end{enumerate}
It is clear that the above scheme has parameters as in (\ref{para}).
We claim that any user $A$ can decode its requested file $W_{d_A}$.
First, all subfiles $W_{d_A,B}$ with $A\cap B\neq \emptyset$ are already in the cache of $A$, so $A$ needs only to retrieve missing subfiles $W_{d_A,B}$ with $A\cap B=\emptyset$.
Fix such a subfile $W_{d_A,B}$.
Put $C=A\cup B$.
In the message $Y_C=\sum_{A'\subset C: |A'|=a}W_{d_{A'},C\setminus A'}$ sent to $A$ by the server, all subfiles $W_{d_{A'},C\setminus A'}, A'\neq A$, are already in the cache of $A$, as $(C\setminus A')\cap A\neq \emptyset$.
Hence $A$ can retrieve the subfile $W_{d_A,C\setminus A}=W_{d_A,B}$.

\medskip

Next, we prove that there is a choice of parameters $a,b,n$ such that the proposed scheme satisfies the conditions (i)-(iii). Put 
\begin{equation}\label{values}
c=\lceil1+1/\epsilon \rceil, \ a= \lceil (\log n)^{c} \rceil, \ b=n-a-c.
\end{equation}
The integer $n$ will be chosen to be big enough and its value is specified later.
Note that $c\geq 1+1/\epsilon$ and $\lim_{n\rightarrow\infty} \frac{n-b}{a}=1$.
To prove (i)-(iii), it suffices to show the following.
\begin{itemize}
\item[(a)] $R\geq R^*$ and $\lim_{n\rightarrow\infty} \frac{R}{R^*}=1$,
\item[(b)] $F\geq K$ for $n$ large enough and $\lim_{n\rightarrow\infty}\frac{F}{K^{(n-b)/a}}<1.$ 
\item[(c)] $\lim_{n\rightarrow\infty}\frac{\log F^*}{(\log F)^c}=\infty$.
\end{itemize}
The proof of (a)-(c) is divided into three claims.

\medskip

\textbf{Claim 1.} $R\geq R^*$ and $\lim_{n\rightarrow\infty} \frac{R}{R^*}=1$.\\
By (\ref{para}), we have $KM/N= {n\choose a}-{n-b \choose a}$. Define 
\begin{equation}\label{inter}
R_0=\frac{K-KM/N}{1+KM/N}=\frac{{n-b \choose a}}{1+{n\choose a}-{n-b \choose a}}.
\end{equation}
Note that by the definition of $R^*$, see (\ref{rate_optimal}), we have $R_0\geq R^*$. Hence to prove $R\geq R^*$, it suffices to show that $R\geq R_0$.
We have 
\begin{equation}\label{rate}
\frac{R}{R_0}=\frac{{n\choose a+b}\left(1+{n\choose a}-{n-b \choose a}\right) }{{n\choose b}{n-b\choose a}}=\frac{{n\choose a}-{n-b \choose a}+1}{{a+b \choose a}}.
\end{equation}
The inequality $R\geq R_0$ is equivalent to
\begin{equation}\label{lower1}
{a+b\choose a}+{n-b\choose a}\leq {n\choose a}+1.
\end{equation}
Viewing $g(b)={a+b\choose a}+{n-b\choose a}$ as a function of $b$ on the interval $[0,n-a]$, we observe that 
$$g(b)\leq g(b+1)\Leftrightarrow b\geq (n-a-1)/2.$$
The function $g(b)$ decreases on the interval $[0,(n-a-1)/2]$ and increases on $[(n-a-1)/2,n-a]$. Thus its maximum is either $g(0)$ or $g(n-a)$. As $g(0)=g(n-a)={n\choose a}+1$, the inequality (\ref{lower1}) follows and we obtain $R\geq R^*$. 

\medskip

 Next, we prove $\lim_{n\rightarrow\infty} \frac{R^*}{R}=1$ by showing that $\lim_{n\rightarrow\infty}\frac{R^*}{R}\geq 1$ (note that we already have $\frac{R^*}{R}\leq 1$ by the previous paragraph). By (\ref{rate_optimal}), we have
 $$R^*\geq \frac{K-KM/N}{1+KM/N}-\frac{{K-1\choose KM/N+1}}{{K\choose KM/N}}=\frac{K-KM/N}{1+KM/N}\left(\frac{1}{K}+\frac{M}{N}\right)=R_0\left(\frac{1}{K}+\frac{M}{N}\right).$$
 It is clear that $\lim_{n\rightarrow \infty}\frac{1}{K}=\lim_{n\rightarrow \infty}{n\choose a}^{-1}=0$.
 Moreover, note that $\frac{M}{N}=1-{n-a\choose b}{n\choose b}^{-1}$ by (\ref{para}) and $\frac{R_0}{R}\geq {a+b\choose a}{n\choose a}^{-1}$ by (\ref{rate}). We obtain
 
 \begin{equation}\label{ineq1} 
 \lim_{n\rightarrow\infty}\frac{R^*}{R}\geq \lim_{n\rightarrow\infty} \frac{{a+b\choose a}}{{n\choose a}}\left(1-\frac{{n-a\choose b}}{{n\choose b}}\right)=\lim_{n\rightarrow\infty} \frac{{a+b\choose a}}{{n\choose a}}\left(1-\frac{{n-a\choose c}}{{n\choose a+c}}\right),
 \end{equation}
On the other hand, by (\ref{values}) we have
 
$$\lim_{n\rightarrow\infty}\frac{a^2}{a+b}=\lim_{n\rightarrow\infty}\frac{a^2}{n}=\lim_{n\rightarrow\infty}\frac{c^2}{n-a}=\lim_{n\rightarrow\infty}\frac{(a+c)^2}{n}=0.$$
Using (\ref{approx}), we obtain

\begin{equation} \label{upper}
 \lim_{n\rightarrow\infty}\frac{{a+b\choose a}}{{n\choose a}}=\lim_{n\rightarrow\infty}\frac{(a+b)^a/a!}{n^a/a!}=\lim_{n\rightarrow\infty} \left(1-\frac{c}{n}\right)^a=\lim_{n\rightarrow \infty}e^{-\frac{ac}{n}}=1
\end{equation}
and
\begin{equation}\label{lowerr}
\lim_{n\rightarrow\infty}\frac{{n-a\choose c}}{{n\choose a+c}}=\lim_{n\rightarrow\infty}\frac{(n-a)^c/c!}{n^{a+c}/(a+c)!}=\lim_{n\rightarrow\infty}\left(1-\frac{a}{n}\right)^c\frac{(c+1)\cdots(c+a)}{n^a}=0,
\end{equation}
where in the last equality, we use 
$$\lim_{n\rightarrow\infty}\left(1-\frac{a}{n}\right)^c=\lim_{n\rightarrow\infty}e^{-ac/n}=1$$ and 
$$\lim_{n\rightarrow\infty} \frac{(c+1)\cdots (c+a)}{n^a}\leq \lim_{n\rightarrow\infty}\left(\frac{a+c}{n}\right)^a=\lim_{n\rightarrow\infty}e^{-ab/n}=0.$$
By (\ref{ineq1}), (\ref{upper}) and (\ref{lowerr}), we obtain
$$\lim_{n\rightarrow\infty}\frac{R^*}{R}\geq 1,$$
finishing the proof of Claim 1.

\medskip

\textbf{Claim 2.} $F\geq K$ for $n$ large enough and $\lim_{n\rightarrow\infty}\frac{F}{K^{(n-b)/a}}<1$.\\
Note that $F={n\choose b}={n\choose n-b}$ and $K={n \choose a}$. As $a\leq n-b<n/2$ for $n$ large enough, we have $F\geq K$ for $n$ large enough. Using Lemma \ref{approx_bin}, we obtain
$$\lim_{n\rightarrow\infty}\frac{F}{K^{(n-b)/a}}=\lim_{n\rightarrow\infty}\frac{n^{n-b}/(n-b)!}{(n^a/a!)^{(n-b)/a}}=\lim_{n\rightarrow\infty}\left( \frac{(a!)^{a+c}}{\left((a+c)!\right)^a}\right)^{1/a}.$$ 
Note that
$$\frac{(a!)^{a+c}}{\left((a+c)!\right)^a}=\frac{(a!)^c}{\left((a+1)\cdots (a+c)\right)^a}<\frac{(a!)^c}{a^{ca}},$$ 
so 
\begin{equation}\label{limit1}
\lim_{n\rightarrow\infty}\frac{F}{K^{(n-b)/a}}\leq \lim_{n\rightarrow\infty} \left(\frac{(a!)^{1/a}}{a}\right)^c.
\end{equation}
By Stirling's approximation formula, we have $\lim_{n\rightarrow\infty}m!\left(\frac{e}{m}\right)^m(2\pi m)^{-1/2}=1$, so $m!<2\sqrt{2\pi m}\left(\frac{m}{e}\right)^m$ for $m$ large enough, which implies
\begin{equation}\label{limit2}
\frac{(m!)^{1/m}}{m}<\frac{m^{\frac{1}{2m}}}{e}(2\sqrt{2\pi})^{\frac{1}{m}}
\end{equation}
for $m$ large enough. Note that $\lim_{n\rightarrow\infty} a=\infty$. By (\ref{limit1}) and (\ref{limit2}), we obtain 
$$\lim_{n\rightarrow\infty}\frac{F}{K^{(n-b)/a}}\leq  \lim_{n\rightarrow\infty} \left(\frac{a^{\frac{1}{2a}}(2\sqrt{2\pi})^{\frac{1}{a}}}{e}\right)^c=\frac{1}{e^c}<1,$$
proving Claim 2.

\medskip

\textbf{Claim 3.} $\lim_{n\rightarrow\infty}\frac{\log F^*}{(\log F)^c}=0$.\\
Note that $F={n\choose b}$ and $F^*={K\choose KM/N}={{n\choose a}\choose {n-b\choose a}}$.
We will use Lemma \ref{approx_bin} to approximate the fraction $\frac{\log F^*}{(\log F)^c}$. For the approximation of $\log F^*$, observe that
$$0<\frac{{n-b\choose a}^2 }{{n\choose a}}=\frac{(n-b)^2\cdots (n-b-a+1)^2}{n\cdots (n-a+1)a!}\leq \left(\frac{(n-b)^2}{n-a+1}\right)^a\frac{1}{a!}\leq \frac{1}{a}$$
for $n$ large enough (in the last inequality, we use $\lim_{n\rightarrow\infty} \frac{(n-b)^2}{n-a+1}=0$).
Hence $\lim_{n\rightarrow \infty}{n-b\choose a}^2{n\choose a}^{-1}=0$.
By Lemma \ref{approx_bin}, we have

\begin{equation}\label{equal}
\lim_{n\rightarrow\infty}\frac{\log F^*}{(\log F)^c}=\lim_{n\rightarrow\infty} \frac{\log \frac{{n\choose a}^{n-b\choose a}}{{n-b\choose a}!}}{\left(\log\frac{n^{n-b}}{(n-b)!}\right)^c}= \lim_{n\rightarrow\infty}\frac{{n-b\choose a}\log {n\choose a}-\log {n-b\choose a}!}{\left((n-b)\log n-\log(n-b)!\right)^c}.
\end{equation}
Next, we compute the limit in (\ref{equal}) by finding dominating terms in both numerator and denominator, then calculating the ratio of these two terms.
First, considering the denominator, we see that
$$0<\frac{\log (n-b)!}{(n-b)\log n}<\frac{\log (n-b)}{\log n}=\frac{\log (a+c)}{\log n} \leq \frac{\log \left((\log n)^c+1+c\right)}{\log n}.$$
So
\begin{equation}\label{compo3}
\lim_{n\rightarrow\infty}\frac{\log (n-b)!}{(n-b)\log n}=0.
\end{equation}
Next, considering the numerator in (\ref{equal}), we see that 
$$\frac{\log {n-b\choose a}!}{{n-b\choose a}\log {n\choose a}}\leq \frac{\log{n-b\choose a}}{\log {n\choose a}}=\frac{\sum_{i=0}^{a-1} \log\frac{n-b-i}{a-i}}{\sum_{i=0}^{a-1} \log \frac{n-i}{a-i}},$$
which implies
$$0<\frac{\log {n-b\choose a}!}{{n-b\choose a}\log {n\choose a}}\leq \frac{a\log (n-b-a+1)}{a\log \frac{n}{a}}= \frac{\log (c+1)}{\log \frac{n}{\lceil (\log n)^{c}\rceil}}\leq \frac{\log (c+1)}{\log \frac{n}{(\log n)^c+1}}.$$
As $c$ is a fixed integer, we have
\begin{equation}\label{compo4}
\lim_{n\rightarrow\infty}\frac{\log {n-b\choose a}!}{{n-b\choose a}\log {n\choose a}}=0.
\end{equation} 
By (\ref{equal}), (\ref{compo3}) and (\ref{compo4}), we obtain

\begin{equation} \label{inequal}
\lim_{n\rightarrow\infty}\frac{\log F^*}{(\log F)^c}=\lim_{n\rightarrow\infty}\frac{{n-b\choose a}\log {n\choose a}}{(n-b)^c(\log n)^c}=\lim_{n\rightarrow\infty}\frac{\log n^a-\log a!}{c!(\log n)^c},
\end{equation}
where in the last equality, we use (\ref{approx}) to approximate ${n-b\choose a}$ by $(n-b)^c/c!$ and approximate ${n\choose a}$ by $n^a/a!$.
Note that $a=\lceil (\log n)^c\rceil \leq (\log n)^{c}+1$, so $0< \frac{\log a!}{\log n^a}\leq\frac{\log a}{\log n}\leq \frac{\log ((\log n)^{c}+1)}{\log n}$, which implies $\lim_{n\rightarrow\infty}\frac{\log a!}{\log n^a}=0$. 
By (\ref{inequal}), we obtain
$$\lim_{n\rightarrow\infty}\frac{\log F^*}{(\log F)^c}=\lim_{n\rightarrow\infty}\frac{\log n^a}{c!(\log n)^c}=\lim_{n\rightarrow\infty} \frac{\lceil (\log n)^c \rceil}{c!(\log n)^{c-1}}=\infty,$$
proving Claim 3.
\end{proof}

\bigskip

\bigskip

\section{Conclusion} \label{section_conclusion}
In this paper, we study symmetric uncoded caching schemes with low subpacketization levels.
Let $K,M,N$ be parameters of a symmetric uncoded caching scheme, that is, a scheme with $K$ users, $N$ files and user's cache rate $\frac{M}{N}$.
We focus on understanding the trade-off between the rate $R$ and the subpacketization level $F$ of these schemes.
It is known \cite{yu} that the optimal rate $R^*=\frac{K-KM/N}{1+KM/N}-\frac{{K-\min\{K,N\}\choose KM/N+1}}{{K\choose KM/N}}$ can be obtained using the subpacketization level $F^*={K\choose KM/N}$.
However, it is unknown whether we can obtain the same rate with a smaller subpacketization level.
In the case that the answer is no, it is desirable to have a scheme with rate $R$ asymptotically close to $R^*$ and subpaketization level $F$ polynomially large compared to $K$.

\medskip

Our contribution in this paper is to provide answers for the above questions.
Firstly, we prove that if $\frac{M}{N}\leq \frac{\min\{K,N\}}{K}$, then there is no symmetric uncoded caching scheme with rate $R=R^*$ and subpacketization level $F<F^*$.
We also show that in this case, $F\equiv 0 \pmod{F^*}$ is the necessary and sufficient condition for the existence of a symmetric uncoded caching scheme with rate $R=R^*$.
Secondly, we give a detailed analysis for the scheme in \cite[Construction I]{shang} to show that this scheme has rate $R$ asymptotically close to $R^*$, subpacketization level $F$ polynomial in $K$ and $F^*$ sub-exponential in $F$.
To conclude this paper, we propose several open questions in this research direction.

\begin{question}
Construct other schemes with parameters $R,F,K$ satisfying the following conditions.
\begin{itemize}
\item[(i)] $R$ is asymptotically close to $R^*$,
\item[(ii)] $F^*$ is sub-exponentially large compared to $F$, and
\item[(iii)] $F$ is polynomially large compared to $K$.
\end{itemize}
\end{question}

\begin{question}
The scheme proposed in Theorem \ref{mainconstr} has rate $R$ asymptotically close to $R^*$ and $F=K^{1+o(1)}$. Prove or disprove the following statement:
There exists a symmetric uncoded caching scheme with $R$ asymptotically close to $R^*$ and $F$ linear in $K$.
\end{question}

\begin{question}
In Theorem \ref{main_lowerbound}, we proved that there is no symmetric uncoded caching scheme with $R=R^*$ and $F<F^*$, given that the parameters $K,M,N$ satisfy the constraint $\frac{M}{N}\leq \frac{\min\{K,N\}}{K}$. The remaining open case is the case of schemes with $K>N$ and $\frac{M}{N}>\frac{N}{K}$. Is it true that these schemes also cannot have $R=R^*$ and $F<F^*$?
\end{question}

\bigskip

\section*{Appendix: Optimal scheme}
In this section, we provide construction for a symmetric uncoded caching scheme with subpacketization level $F=h{K\choose KM/N}$, where $h$ is a positive integer, and rate 
$$R=\frac{K-KM/N}{1+KM/N}-\frac{{K-\min\{K,N\} \choose KM/N+1}}{{K\choose KM/N}}.$$
The scheme is designed as follows.
\begin{enumerate}
\item Let $W_1,\dots,W_N$ denote $N$ files. Put $t=KM/N$. Each $W_i$ is partitioned into $F$ subfiles $W_{i,j,S}$, where $j\in [h]$ and $S$ is a subset of $[K]$ with $|S|=t$.
\item The users are denoted by numbers $1,\dots,K$. User $u$ stores $W_{i,j,S}$ in its cache if and only if $u\in S$. In this way, each user $u$ stores $Z=h {K-1\choose t-1}$ subfiles of each file $W_i$. 
\item Let $d=(d_1,\dots,d_K)$ be a demand from the users such that user $u$ requests $W_{d_u}$. The delivery proceeds as follows.
\begin{itemize}
\item Let $e$ be the number of distinct files from the set $\{W_{d_1},\dots,W_{d_K}\}$ and assume $U=\{i_1,\dots,i_e\}\subset [K]$ is a set of users requesting these $e$ files. 
\item For any $j\in [h]$ and any $A\subset [K]$ with $|A|=t+1$ and $A\cap U\neq \emptyset$, the server sends
$$Y_{j,A}=\sum_{i\in A} W_{d_i,j,A\setminus \{i\}}.$$
\end{itemize}
\end{enumerate}

\noindent The number of files sent in the above scheme is 
$$h\left( {K\choose t+1}-{K-e \choose t+1} \right)\leq h\left( {K\choose t+1}-{K-\min\{K,N\} \choose t+1} \right),$$
as $e\leq \min\{K,N\}$.
The equality happens when the users request $e=\min\{K,N\}$ distinct files.
Thus the scheme has rate
$$R =\frac{h\left( {K\choose t+1}-{K-\min\{K,N\} \choose t+1} \right)}{F}= \frac{K-KM/N}{1+KM/N}-\frac{{K-\min\{K,N\} \choose KM/N+1}}{{K\choose KM/N}}.$$
It remains to show that any user $u$ can decode its requested message $W_{d_u}$.
As $u$ already has the subfiles $W_{d_u,j,S}$ with $S\ni u$ in its cache, it only needs to recover the missing subfiles $W_{d_u,j,S}$ with $u\not\in S$.
This can be done if user $u$ knows all messages 
$$Y_{j,B}=\sum_{i\in B}W_{d_i,j,B\setminus\{i\}}, \ B\subset \{1,\dots,K\} \ \text{and} \ |B|=t+1.$$
Indeed, let $W_{d_u,j,S}$, $S\not\ni u$, is a subfile not in the cache of $u$. Put $B=S\cup\{u\}$.
As user $u$ knows $Y_{j,B}$ and has all subfiles $W_{d_i,j,B\setminus\{i\}}$, $i\neq u$, in its cache (note that $B\setminus \{i\}\ni u$ for any $i\neq u$), user $u$ can retrieve the subfile $W_{d_u,j,S}$. 

\medskip

Now we prove that all $Y_{j,B}$ are known by all users.
As the server sends directly all $Y_{j,B}$ with $B\cap U\neq \emptyset$ to the users, the unsent ones are $Y_{j,B}$ with $B\cap U=\emptyset$.
Fix $B\subset \{1,\dots,K\}$ such that $|B|=t+1$ and $B\cap U=\emptyset$.
Put $C=B\cup U$. Let $\mathcal{V}$ be the set of all $e$-subsets $V$ of $C$ such that the users in $V$ request all $e$ distinct files from $\{W_{d_1},\dots,W_{d_K}\}$.
Note that $U\in \mathcal{V}$.
The message $Y_{j,B}$ is obtained by the following equation whose proof is from \cite[Lemma 1]{yu}.
\begin{equation}\label{retrival}
\oplus_{V\in \mathcal{V}}Y_{j,C\setminus V}=0.
\end{equation}
For any $V\neq U$, we have $(C\setminus V)\cap U\neq\emptyset$, so the message $Y_{j,C\setminus V}$ is sent directly by the server. Thus $Y_{j,B}$ is the only unknown component in (\ref{retrival}) and its value can be obtained from (\ref{retrival}).


\bigskip

\begin{thebibliography}{10}

\bibitem{alm} K. C. Almeroth, M. H. Ammar:
The use of multicast delivery to provide a scalable and interactive video-on-demand service, \textit{IEEE J. Sel. Areas Communi.}, 14 (1996), 1110 -- 1122.

\bibitem{bae} I. Baev, R. Rajaraman, C. Swamy:
Approximation algorithms for data placement problems, \textit{SIAM J. Comput.}, \textbf{38} (2008), 1411 -- 1429.

\bibitem{bha} M. Bhavana, H. H. S. Chittoor, P. Krishnan: Coded Caching via Projective Geometry: A new low subpacketization scheme. \textit{Proc. IEEE Int. Symp. Inf. Theory (ISIT)}, 2019, 682 -- 686.

\bibitem{bor} S. Borst, V. Gupta, A. Walid: 
Distributed caching algorithms for content distribution networks, \textit{Proc. IEEE Int. Conf. Computer Communi.} (INFOCOM), 2010, 1478 -- 1486.

\bibitem{chi} H. H. S. Chittoor, P. Krishnan:
Low Subpacketization Coded Caching via Projective Geometry for Broadcast and D2D networks. \textit{IEEE Global Communications Conference} (GLOBECOM) 2019.

\bibitem{dow} L. W. Dowdy, D. V. Foster: 
Comparative models of the file assignment problem, \textit{ACM Comput. Surv.}, \textbf{14} (1982), 287 -- 313.

\bibitem{ali1} M. A. Maddah-Ali, U. Niesen: 
Fundamental limits of caching, \textit{IEEE Trans. Info. Theory}, \textbf{60} (2014), 2856 -- 2867. 

\bibitem{nie} U. Niesen, M. A. Maddah-Ali:
Coded caching with nonuniform demands, \textit{IEEE Trans. Inf. Theory}, \textbf{63} (2017), 1146 -- 1158.

\bibitem{pia} P. Piantanida, D. Tuninetti, K. Wan:
On the optimality of uncoded cache placement, \textit{Proc. IEEE Inf. Theory Workshop (ITW)}, 2016.  

\bibitem{ram} A. Ramamoorthy, L. Tang: 
Coded Caching with Low Subpacketization Levels, \textit{Proc. Globecom Workshops (GC Wkshps)}, 2016, 1 -- 6.

\bibitem{shang} C. Shangguan, Y. Zhang, G. Ge:
Centralized Coded Caching Schemes: A Hypergraph Theoretical Approach. \textit{IEEE Trans. Info. Theory}, \textbf{64} (2018), 5755 -- 5766. 

\bibitem{sha} K. Shanmugam, M. Ji, A. M. Tulino, J. Llorca, A.G. Dimakis:
Finite-length analysis of caching-aided coded multicasting, \textit{IEEE Trans. Inf. Theory}, \textbf{62} (2016), 5524 -- 5537.

\bibitem{tia} C. Tian, J. Chen: 
Caching and Delivery via Interference Elimination.  \textit{IEEE Trans. Inf. Theory}, \textbf{64} (2018), 1548 -- 1560.

\bibitem{yan1} Q. Yan, M. Cheng, X. Tang, Q. Chen: On the placement delivery array for centralized coded caching scheme.  \textit{IEEE Trans. Info. Theory}, \textbf{63} (2017), 5821 -- 5833. 

\bibitem{yan2} Q. Yan, M. Cheng, X.Tang, Q. Chen: Placement delivery array design through strong edge coloring of bipartite graphs. \textit{IEEE Commun. Lett.}, \textbf{22} (2018), 236 -- 239.

\bibitem{yu} Q. Yu, M. A. Maddah-Ali, A. S. Avestimehr: The exact rate-memory tradeoff for caching with uncoded prefetching, \textit{IEEE Trans. Info. Theory}, \textbf{2} (2018), 1281 -- 1296. 

\bibitem{yu2} Q. Yu, M. A. Maddah-Ali, A. S. Avestimehr:
Characterizing the rate-memory tradeoff in cache networks within a factor of $2$. \textit{IEEE Trans. Info. Theory}, \textbf{65} (2019), 647 -- 663. 

\end{thebibliography}
\end{document}